\documentclass{llncs}

\usepackage[pdftex]{graphicx}
\usepackage[cmex10]{amsmath}

\usepackage{amssymb}
\usepackage{algorithm2e}
\usepackage{enumerate}

\begin{document}
\title{Approximate Congestion Games for Load Balancing in Distributed Environment}

\author{Sandip Chakraborty, Soumyadip Majumder, Diganta Goswami
\thanks{A version of this paper has been presented at International Workshops on Distributed Systems, 2010, held at IIT Kanpur, India, as  ``work-in-progress'' document.}}
\institute{Department of Computer Science and Engineering\\
Indian Institute of Technology, Guwahati\\
Guwahati, India\\
Email: \{c.sandip, s.majumder, dgoswami\}@iitg.ernet.in
}

\maketitle

\begin{abstract}
The use of game theoretic models has been quite successful in describing various cooperative and non-cooperative optimization problems in networks and other domains of computer systems. In this paper, we study an application of game theoretic models in the domain of distributed system, where nodes play a game to balance the total processing loads among themselves. We have used congestion gaming model, a model of game theory where many agents compete for allocating resources, and studied the existence of Nash Equilibrium for such types of games. As the classical congestion game is known to be PLS-Complete, we use an approximation, called the $\epsilon$-Congestion game, which converges to $\epsilon$-Nash equilibrium within finite number of steps under selected conditions. Our focus is to define the load balancing problem using the model of $\epsilon$-congestion games, and finally provide a greedy algorithm for load balancing in distributed systems. We have simulated our proposed system to show the effect of $\epsilon$-congestion game, and the distribution of load at equilibrium state.
\end{abstract}

\section{Introduction}
One of the most important problem to attain high performance in distributed system with heterogeneous clients is to balance the total loads of all the processing jobs among the whole system, such that the processing burdens of all clients are almost same. This is the classical problem of load balancing in distributed system. A distributed system can be considered as a collection of computing devices connected with communication links by which resources are shared among active users.  In {\cite{grosu}}, the authors describe the load balancing problem as follows \textit{``Given the initial job arrival rates at each computer in the system find an allocation of jobs among the computers so that the response time of the entire system over all jobs is minimized"}. This definition of load balancing problem can be viewed as the \textit{Static Load Balancing}, where the main assumption is that all information governing load-balancing decisions such as characteristics of jobs, the computing nodes and the communication links are knows at advance. Here load-balancing decisions are made deterministically or probabilistically at compile time and remains constant during runtime. A variant of classic load balancing problem is the \textit{Dynamic Load Balancing}, where load balancing decisions are made at run time, and according to the current load of different computing devices, loads are transferred from one computing device to another dynamically at run time. In {\cite{sharma}}, the authors have shown that static load balancing algorithms are more stable in compare to dynamic load balancing algorithms and it is easy to predict the behavior of static algorithms. In this paper we model such static load balancing algorithm for distributed system using congestion game model. 
\begin{sloppypar}
The problem of load balancing in distributed system can be handled in three ways:
\begin{itemize}
	\item {\textit{Global Approach}: A single decision maker optimizes the overall response time using some optimization techniques. This approach is also called as \textit{social optimum}. This is the most frequent literature in study and has been studied extensively using different techniques such as nonlinear optimization, polynomial optimization etc.}
	\item {\textit{Cooperative Approach}: This is based on classical cooperative game theory where the decision makers cooperate between themselves by sharing informations through message passing, and then take the decision using the utility and pay-off function. In {\cite{grosu}}, the authors have modeled the problem of load balancing in distributed system using cooperative game theoretic approach.}
	\item{\textit{Non-cooperative Approach}: Here the decision is made using the pay-off and utility function, but the agents does not share any information between themselves. Congestion game is a variant of non-cooperative games where each agent's strategy consists of a set of resources, and the cost of the strategy depends only on players using each resource.}
\end{itemize}
\end{sloppypar}
\begin{sloppypar}
Congestion games have attracted a good deal of attention, because they can model a large class of routing and resource allocation scenarios. Another reason for using congestion games in resource allocation is that they possess \textit{pure Nash equilibria} {\cite{vocking}}. In general games, Nash equilibrium may involved mixed ( i.e., randomized) strategies for players, but congestion game always have a Nash equilibrium in which each player sticks to a single strategy. But the problem with Nash equilibrium in congestion game is that they are known to be \textit{PLS-Complete}{\cite{fabrikant}}. So it is difficult to find a Nash equilibrium in congestion games and so the convergence time is very high. In {\cite{chien}}, the authors describe a variant of congestion game, called $\epsilon$-congestion game which is an approximate version of pure Congestion games, and have been proved to possess a better convergence rate. The authors have proved that $\epsilon$-Nash dynamics in $\epsilon$-congestion game converge to an $\epsilon$-Nash equilibrium within a finite number of steps. 
\end{sloppypar}
\begin{sloppypar}
We have used the formulation of $\epsilon$-Nash Equilibrium, as described in {\cite{chien}}, to model the problem of load balancing in distributed system as a resource sharing problem. We have studied the formulation of the problem and the existence of $\epsilon$-Nash equilibrium for the problem.
\end{sloppypar}

\section{Related Works}
The problem of static load balancing for single class job distributed system has been studied extensively for global approach. In global approach the focus is to minimize the overall response time. In {\cite{tantawi}}, the authors have formulated the load balancing problem as a nonlinear optimization problem and have given an algorithm to solve that nonlinear optimization problem. In {\cite{kim}} and {\cite{kim2}}, Kim and Kameda derived a more efficient algorithm to introduce the problem. An comparison between several static load balancing algorithms with respect to job dispatching strategy has been studied by Tang and Chanson, in {\cite{tang}}.
\begin{sloppypar}
	The load balancing problem using game theoretic approaches also has been studied both for cooperative and non-cooperative approach. In {\cite{altman}}, Altman, Kameda and Hosokawa modeled distributed system as collection of nodes connected by either simplex or duplex connected links, and then described the dynamic load balancing problem for this system using game theoretic approach. They have established the proof for a unique Nash equilibrium for routing games with the above mentioned distributed system model, under quite general assumption on the costs. For this, they have considered two different architectures, in one the nodes are connected using duplex link, and in another they are connected via two one-way communication links. In {\cite{grosu}}, the authors have modeled the system as an M/M/1 queue, and then proposed an algorithm, called \textit{``Cooperative Static Scheme"}(COOP), using Nash Bargaining Solution, an interesting variant of cooperative game theoretic concepts and the solution of the problem using first order Kuhn-Tuker conditions. They have also compared their Nash bargaining solution algorithm with the existent static and dynamic load balancing algorithms.There are very few literatures in studying the non-cooperative models for load balancing problems in distributed systems. Kameda \textit{et al.} {\cite{kameda}} studied non-cooperative games and derived load balancing algorithms for both single class and multi class job distributed systems. For single class job, they have proposed an algorithm for computing Wardrop equilibrium, a variant of Nash equilibrium where number of agents participating in the game is infinite. In {\cite{roughgarden}}, the author has modeled the load balancing problem as a Stackelberg game, where one player acts as a leader and rest as followers. He has showed that optimal Stackelberg Strategy computing is NP-hard and hence he formulated a near optimal solution. In {\cite{suri2}}, the authors has proposed an uncoordinated load balancing algorithm for peer-to-peer system. They have analyzed the Nash Equlibrium under general latency function for a P2P system.
\end{sloppypar}
\begin{sloppypar}
	In this paper we have moved from these general game theoretic approaches, and modeled the problem as a congestion game, the most suitable strategy to model resource allocation problems using non-cooperative game models. The work of {\cite{fabrikant}} shows that finding a pure Nash equilibrium is PLS-Complete, and hence the convergence time for congestion games, that always produces a pure Nash equilibrium, is very high. Again, for some initial strategies the shortest path to an equilibrium in the Nash dynamics is exponentially long according to the the number of players. In {\cite{suri}}, the authors has proposed a selfish load balancing problem using atomic congestion game, and have shown that the worst case ratio between a Nash Solution and a Social Optimum, which is referred several times as Price of Anarchy, is at most 2.5. A recent advance in the field of congestion gaming model is $\epsilon$-congestion game. In {\cite{chien}}, the authors have studied these issues for congestion games, and proposed an approximation of pure congestion game, which they named as $\epsilon$-congestion game, where $\epsilon$-Nash dynamics converges to an $\epsilon$-Nash equilibrium within a finite number of steps under some bounded conditions. This work on the approximation on pure congestion game is the primary motivation behind our work to check how the load balancing problem can be well suited with the approximated congestion game, or the $\epsilon$-congestion game. We have used this version of congestion game to model our load balancing scheme. Finally we have shown by simulation that using a proper, well chosen value of $\epsilon$, we can reduce the number of iteration for congestion game to a very lower bound, and at this bound the load is well distributed among the processing nodes to minimize the processing time. 
\end{sloppypar}
\section{Concept of Congestion Game}
In this paper we address the problem of load balancing in the general arena of \textit{Congestion Game}, or more specifically an approximation of classical congestion game, called \textit{$\epsilon$-Congestion Game}. A congestion game can be formally described as a finite set of players $\left\{p_1,p_2,...,p_n\right\}$, each of which is assigned a finite set of \textit{strategies} $S_i$ and a cost function $c_i:S_1{\times}S_2...{\times}S_n{\longrightarrow}{\mathbb{N}}$ that he wishes to minimize. A \textit{state} $s = (s_1,s_2,...,s_n) \in S_1 \times S_2 \times ... \times S_n$ is any combination of strategies for the player. A state s is a \textit{pure Nash Equilibrium} if for all players $p_i$, $c_i(s_1,s_2,...,s_i,...,s_n) \leq c_i(s_1,s_2,...,\acute{s_i},...,s_n)$ for all $\acute{s_i} \in S_i$. Thus we can say that at a pure Nash equilibrium, no player can improve his cost by unilaterally changing his strategy. It is well known that every finite game has a mixed Nash equilibrium, but not a pure Nash equilibrium. Whereas congestion game always has pure Nash equilibrium.
\begin{sloppypar}
	In case of congestion games, players' cost are based on the shared usage of a common set of resources ( also called edges, in terms of network congestion games) R = $\{r_1,r_2,...,r_m\}$. A player's strategy set $S_i \subseteq 2^R$ is an arbitrary collection of subsets of R; his strategy $r_i$ will be therefore a subset of R. Each resource $r \in R$ has an associated \textit{nondecreasing} delay function $d_r : \{1,2,...,n\} \longrightarrow \mathbb{N}$. If t players are using the resource r, they will each incur a cost of $d_r(t)$. As a result in a state s = $(s_1,s_2,...,s_n)$, the cost of player $p_i$ is $c_i(s) = \sum\limits_{r \in s_i}^{}{d_r(f_s(r))}$, where $f_s(r)$ is the number of players using resource r under s; i.e. $f_s(r) = \left|\{j : r \in s_j\}\right|$ .
\end{sloppypar}
\begin{sloppypar}
	In {\cite{rosenthal}}, the author has shown that in every congestion game,every sequence of improvement steps is finite. This proposition can be shown by a potential function argument, called the \textit{Rosenthal's Potential function} $\phi : \sum_1\times...\times\sum_n\longrightarrow\mathbb{Z}$, defined as
	\[
	\phi{(s)}=\sum_{r \in R}{\sum\limits_{i=1}^{f_s(r)}{d_r(i)}}
\] 
This function has an interesting property that if player $p_i$ shifts strategy from $s_i$ to $\acute{s_i}$, the change in $\phi$ exactly mirrors the change in the player's cost: i.e., $\phi{(s)} - \phi{(\acute{s})} = c_i(s) - c_i(\acute{s}) $. The consequence from this observation is that, if we follow an iterative process here, such that, at each step one player changes its strategy to lower his cost ( a Nash dynamics), then the potential function $\phi$ will decrease until it reaches a local minimum, which must be a pure Nash equilibrium. But this does not provide a bound on the number of such player moves required to reach a pure Nash equilibrium in a congestion game ( so the congestion game is PLS-Complete), and this problem leads to the concept of a approximate Nash equilibrium, which is $\epsilon$-Nash equilibrium.
\end{sloppypar}
\section{$\epsilon$-Nash Dynamics and $\epsilon$-Nash Equilibrium}
In {\cite{chien}}, the authors has studied and formulated the concept of $\epsilon$-Nash Dynamics and $\epsilon$-Nash equilibrium. A state s  is an $\epsilon$-Nash equilibrium if no player can improve his cost by more than a factor of $\epsilon$ by unilaterally changing his strategy. The $\epsilon$-Nash dynamics is a modification of pure Nash dynamics, where players are permitted only $\epsilon$-moves, i.e. moves that improve the cost of the player by a factor of more than $\epsilon$. To make $\epsilon$-Nash dynamics concrete, the authors assume that among multiple players with $\epsilon$-moves available, at each step a move is made by the player with the largest incentive to move; i.e. the player who can make the largest relative improvement in cost. For this reason they have also introduced $\alpha$-bounded jump, where the delay function satisfies the condition $d_r(t+1) \leq \alpha{d_r(t)}$ for all $t \geq 1$ and $\alpha \geq 1$. In particular, a resource ( or edge) with $d_r(t)=\alpha^t$ satisfies the $\alpha$-bounded jump condition. They show that in such a condition, $\epsilon$-Nash dynamics converge to an $\epsilon$-Nash equilibrium within a finite number of steps, in fact in $\left\lceil n\alpha\epsilon^{-1}\log{(nC)}\right\rceil$ steps, where C is an upper bound on the cost of any player. This is the main motivation behind our work to use this modified approximate congestion gaming model to solve the load balancing problem in distributed system.We have the following theorem about $\epsilon$-Nash equilibrium {\cite{chien}}.
\begin{theorem}
\label{theorem:epsilonnash}
For $\epsilon \in \left[0,1\right)$, a state s = $(s_1,...,s_n) \in S_1\times...\times S_n$ is an $\epsilon$-Nash equilibrium if for all player $p_i$, $c_i(s_1,...,\acute{s_i},...,s_n) \geq (1-\epsilon)c_i(s_1,...,s_i,...,s_n)$ for all $\acute{s_i} \in S_i$.
\end{theorem}
\section{System Model}
In a distributed system, some nodes act as the processor nodes, who process the jobs, and some nodes act as the load generator, who generates processes or jobs in a certain rate. We call the processor nodes as the servers, and the load generators as the client. It should be noted that in a typical distributed system, same node can act both as server and clients. But for the simplicity of modeling the system, we consider server and clients as separate nodes without the loss of generality, and our target is to assign the clients to the servers in a balanced way, such that the jobs generated at the clients are distributed almost equally among the servers and the overall system response time is minimized. Hence we need to generate a bipartite graph with n servers and m clients as shown in the figure {\ref{fig:systemmodel}}.
\begin{figure}[!t]
    \centering
        \includegraphics[width=3in]{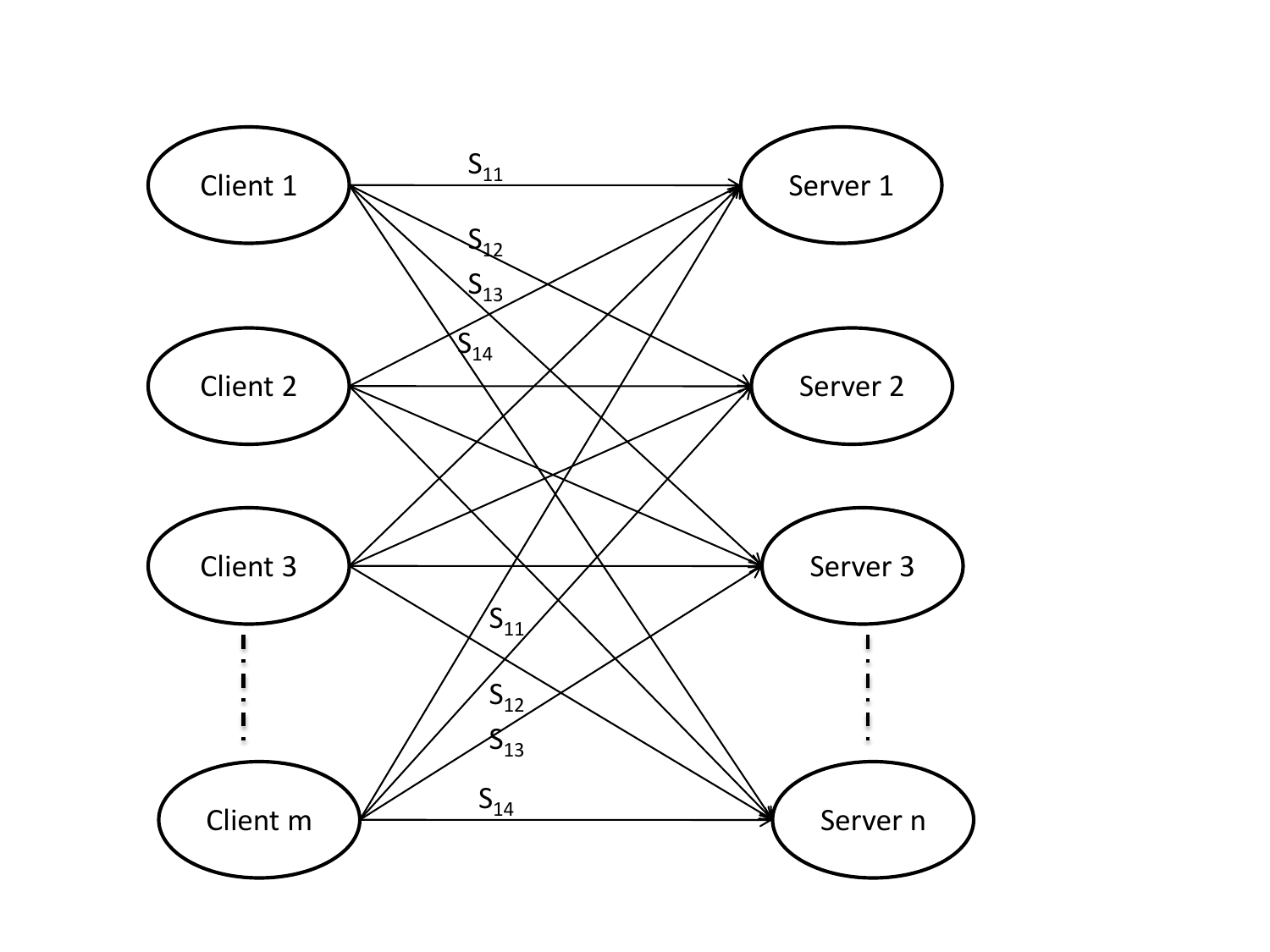}
    \caption{Distributed System Model for Load Balancing Problem}
    \label{fig:systemmodel}
\end{figure}
We use following parameters to model our system:\\
$\mu_i^{max}$ = The maximum processing rate of server i.\\
$\mu_i$ = The actual or effective processing rate of server i.\\
$\phi_j$ = Load generation rate at client j.\\
$S_{ji}$ = Fraction of the load assigned to server i by client j.\\
$\phi$ = $\sum\limits_{j=1}^{m}{\phi_j}$ which is total job arrival rate in complete distributed system.\\ 
Here, j=1,2,...,m and i=1,2,...,n.
It is obvious that for each client j;\\
$\sum\limits_{i=1}^{n}{S_{ji}}$ = 1; where $0\leq S_{ji}\leq 1$ ; and i=1,2,...,m\\
Our objective is to find the fraction $S_{ji}$ for each client j (j = 1,...,m) such that the expected execution time of the jobs in the system is minimized. We have the following conditions;\\
At each server i;\\
$\sum\limits_{j=1}^{m}{S_{ji}\phi_j} < \mu_i$; which implies that actual processing rate at any instance at server i must be less than the processing rate capacity or the advertised processing rate of server i. \\
Modeling each computer as an M/M/1 queuing system {\cite{luenberger}},\\
$F_i(\beta_i) = \frac{1}{\mu_i - \beta_i}$, where $\beta_i$ = average arrival rate of jobs at server i and $F_i(\beta_i)$ denotes the expected execution time of jobs processed at server i. For M/M/1 queueing system, the above condition must be satisfied also which guarantees stability of the overall system.\\
In our case, $\beta_i = \sum\limits_{j=1}^{m}{S_{ji}\phi_j}$ \\
Hence, $F_i(\beta_i) = \frac{1}{\mu_i - \sum\limits_{j=1}^{m}{S_{ji}\phi_j}}$\\
We can calculate the overall response time of client j as follows:\\
$R_j(\beta_i) = \sum\limits_{i=1}^{n}{S_{ji}F_i(\beta_i)} = \sum\limits_{i=1}^{n}{\frac{S_{ji}}{\mu_i-\sum\limits_{k=1}^{m}{S_{ki}{\phi_k}}}}$\\
Our objective is to minimize the overall response time of the system.
\section{Load Balancing as a $\epsilon$-Congestion Game among Clients}
Let for client j, the vector $S_j$ = $(S_{j1},S_{j2},...,S_{jn})$ be the load balancing strategy of client j. Then the set S = $\{S_1,S_2,...,S_m\}$ corresponds to the strategy space or strategy profile of our load balancing game. The cost function of client i is $c_i=R_i(\beta_j)$, and the corresponding delay function at resource i is $d_r(t)=F_i(\beta_i)$ where t is the number of resources, that is number of servers in load balancing game. It is clear that the delay function depends on number of clients associated with the particular resource and the function is nondecreasing. At this point of instance, we have the following assumptions:

\begin{enumerate}[(i)]
\item \textit{Each server i can tolerate up to a maximum job arrival rate $\lambda_i^{max}$.} Though in theory each server is capable to process jobs up to its processing rate, but in practice the performance of the system drops dramatically after a specific amount of load, because each server has to process some internal loads also. So this assumption is more likely to the real world scenario and required to proof the convergence of our $\epsilon$-congestion game.
\item \textit{There is a maximum job generation rate at the system $\phi_{max}$}. This assumption is also near to reality.
\end{enumerate}

The convergence property for our load balancing $\epsilon$-congestion game can be shown using following theorem {\ref{theorem:boundedjump}}. 
\begin{theorem}
\emph{\textbf{($\alpha$-bounded Jump Condition)}}
\label{theorem:boundedjump}
The delay function for resource j, as given above is nondecreasing and satisfies the $\alpha$-bounded jump condition for $\epsilon$-congestion game.
\end{theorem} 
\begin{proof}
For resource $i$,\\
$d_r(t)=\frac{1}{\mu_i-\sum\limits_{j=1}^{t}{S_{ji}\phi_j}}$, and\\
$d_r(t+1)=\frac{1}{\mu_i-\sum\limits_{j=1}^{t+1}{S_{ji}\phi_j}}$\\
Hence,$\frac{d_r(t+1)}{d_r(t)} = \frac{\mu_i-\sum\limits_{j=1}^{t}{S_{ji}\phi_j}}{\mu_i-\sum\limits_{j=1}^{t+1}{S_{ji}\phi_j}}$\\
This can be simplified as,\\
$\frac{d_r(t+1)}{d_r(t)} = 1 + \frac{S_{(t+1)i}\phi_{t+1}}{\mu_i-\sum\limits_{j=1}^{t+1}{S_{ji}\phi_j}}$\\
So, it is clear that $d_r(t+1) \geq d_r(t)$, so the delay function is nondecreasing.
According to our assumptions, for each server j, we can get $\alpha_j$ as;\\
$\alpha_j = 1 + \frac{\phi_{max}}{\mu_i - \lambda_i^{max}}$\\
clearly, $\alpha_j > 1$ for all server ( resources ) j. Hence the $\alpha$-bounded jump condition is satisfied. So in this game, the $\epsilon$-Nash dynamics will converge to $\epsilon$-Nash equilibrium within a finite number of steps.
\end{proof}
We define the potential function for our game which is similar to Rosenthal's Potential Function, but with a multiplier in addition with the delay function. The multiplier comes because the overall response time of all the servers depends on a fraction of each client's load. The potential function $\phi(s)$ with strategy set s of the system is as follows :\\
$\phi(s) = \sum\limits_{j=1}^{n}{\sum\limits_{i=1}^{m}{S_{ij}{\frac{1}{\mu_j - \sum\limits_{k=1}^{m}{S_{kj}\phi_k}}}}}$ where m is the number of clients (players) and n is the number of servers (resources). It is clear that the overall system response time depends only on those fractions for which the multiplication factor $S_{ij}$  is nonzero. The cost function incurred at each client i with strategy set s is as given before;\\
$c_i(s) = \sum\limits_{i=1}^{n}{\frac{S_{ji}}{\mu_i - \sum\limits_{k=1}^{m}{S_{ki}\phi_k}}}$
\begin{theorem}
\emph{\textbf{(Existence of Exact Potential Function)}}
The potential function as defined above obeys the property of Rosenthal's Potential Function, i.e. if player $p_i$ shifts strategy from $s$ to $\acute{s}$, the changes in $\phi$ exactly mirrors the change in the player's cost. 
\end{theorem}
\begin{proof}
As $\sum\limits_{i=1}^{n}{S_{ji}} = 1$ for each client j, so the total weighted sum of the multipliers at both the cost function and potential function is always 1, which is a constant value, and thus it does not incurs any change in the relative deference between potential functions and costs with two strategies $s$ and $\acute{s}$. At each independent move, the change in the player's cost is equal to the change in total delay at all resources, as the delay at each resources reflects directly at the player's cost.  Hence from the structure of both the functions it is easy to conclude that, $\phi(s) - \phi(\acute{s}) = c_i(s) - c_i(\acute{s})$, which follows the theorem.
\end{proof}
It can be also checked easily that the cost function satisfies Theorem {\ref{theorem:epsilonnash}}, that is the cost can be increased always above a threshold given by the value of $\epsilon$. Hence it is clear that our load balancing game model converges to $\epsilon$-Nash equilibrium within a finite steps. Our objective is to minimize the potential function starting from a initial strategy $s_0$, where each node independently tries to minimize its own cost function and that directly reflects the change in overall system response time which can be measured using the potential function as defined above.
\begin{sloppypar}
	The optimization problem at each client i will be as follows;\\
	$min_{S_j}R_j(s)$ where $S_j$ is the current strategy profile for client j and s is the strategy set of the system.\\
	Subjected to;\\ \\
	$S_{ji} \geq 0$, i=1,2,...,n\\
	$\sum\limits_{i=1}^{n}{S_{ji}} = 1$\\
	$\sum\limits_{k=1}^{n}{S_{ki}\phi_k} < \mu_i$, i=1,2,...,m\\
	\end{sloppypar}
	\begin{sloppypar}
	From the above conditions it is clear that the servers with higher processing power should have higher fractions of jobs assign to it. If the computers are arranged in the decreasing order of their processing rates, then we get a partial order for client j as $S_{j1} > S_{j2} > ... > S_{jn}$.
	Now in real system, there will be some servers with low processing rates, where no load will be assigned. So after a index k, $S_{ji} = 0$ for i=k,k+1,...,n.
\end{sloppypar}
Now when a client j runs its own optimization problem, the algorithm is given in (Algorithm {\ref{algo:optimal}}).Each client node runs this algorithm independently to find out its own local optimum solution.
\begin{algorithm}[!t]
\textbf{Algorithm OPTIMAL:}\\
\KwIn {$\mu_1,\mu_2,...,\mu_n,\phi_j$}
\KwOut {$S_j$, the strategy set for user j}
Arrange the computers in decreasing order of their processing rates, $\mu_1,\mu_2,...,\mu_n$\\
let $temp \leftarrow \sum\limits_{i=1}^{n}\mu_i - \phi_j$ \\
\While{temp $\geq \mu_n$}{
			$S_{jn} \leftarrow 0$\;
			$n = n - 1$\;
			$temp \leftarrow \sum\limits_{i=1}^{n}\mu_i - \phi_j$\;
}
$sum = \sum\limits_{i=1}^{n}{\mu_i}$\;
\ForEach{$i \in \{1,2,...,n\}$}{
			${loadFrac} = \mu_i / sum$\;
			\If{canAssign($loadFrac$)}{
						$S_{ji} \leftarrow loadFrac$\;
			}
			\Else{
						$S_{ji} \leftarrow 0$\;
						$sum = sum - \mu_i$\;
			}
}
\caption{Calculating Optimal Solution for Client j with Current Effective Processing Rates $\mu$ of Servers}
\label{algo:optimal}
\end{algorithm}
\begin{sloppypar}
	For the overall system optimization, we assume that the clients enter the system one after another. This assumption follows directly from the properties of potential function. The initial strategy set of the system is $S_0$, and according to the property of $\epsilon$-Nash equilibrium, the system converges within a finite number of steps with any initial strategy. One initial strategy can be that whole job of a node is assigned to its own server. With this initial strategy the system moves to a new state as (Algorithm {\ref{algo:cong}}). The algorithm works as a greedy strategy where each node chooses the current best solution and according to the greedy property and reachability conditions of congestion games, the system guarantees to be terminated at $\epsilon$-Nash Equilibrium within a finite number of steps. Here the system designer needs to choose $\epsilon$ according to the system design using practical simulation.
	\begin{algorithm}[!t]
	\textbf{Algorithm LOAD\_BALANCE:}\\
	\KwIn{$S_0$, initial strategy set, and $\mu=\{\mu_1,\mu_2,...,\mu_n\}$, the current effective processing rate of servers}
	\KwOut{$S_{eq}$, the strategy at Nash Equilibrium}
	$x=1$\;
	\ForEach{client j at iteration x}{
		\ForEach{Server i}{
			Initialize $\mu_i \leftarrow \mu_i^{max}$
		}
		\ForEach{Server i}{
			$\mu_i^x \leftarrow \mu_i^{(x-1)} - \sum\limits_{k=1; k \ne j}^{m}{S_{ki}\phi_k}$ \;
		}
		$S_j^x = OPTIMAL(\mu_1,\mu_2,...,\mu_n,\phi_j)$\;
		Compute $c_j^x(S_j^x)$, the cost of user j with strategy profile $S_j^x$\;
		\If{$c_j^x(S_j^x) - c_j^{(x-1)}(S_j^{(x-1)}) < \epsilon$}{
			$S_j^x = S_j^{(x-1)}$\;
			$c_j^x = c_j^{(x-1)}$\;
		}
		$x=x+1$	
	}
	Exit when there is no change in cost, that is Nash equilibrium has been reached.\;
	\caption{Load Balancing using $\epsilon$-Congestion Game}
	\label{algo:cong}
	\end{algorithm}  
\end{sloppypar}
\section{Simulation Results}
We have simulated the system using a program written in Java programming language. Each server has some maximum processing rates, and we use the random number generator to generate load at each clients. The unit of the processing rates of the server and the job generation rate at client is taken as number of jobs per second. For our simulation we have taken different situation with different number of servers with processing rates varying from small to large amount, and similarly for clients. We have maintained the dynamic nature of the system as much as possible while setting up the simulation environment.
\begin{sloppypar}
	To show that the load is balanced among the processing nodes, we have defined the metrics \textbf{Load Ratio}.The metrics is defined as follows,\\ \\
	${LR_i} = \frac{L_i}{\mu_i^{max}}
								= \frac{\mu_i^{max} - \mu_i}{\mu_i^{max}}$;\\
	where,\\ $LR_i$ = Load Ratio at server i,\\
	$L_i$ = Total Load assigned at server i, \\
	$\mu_i^{max}$ = Maximum processing rate of server i, \\
	$\mu_i$ = Current Effective processing rate of server i.\\
\end{sloppypar}
\begin{sloppypar}
\begin{figure}[!t]
    \centering
        \includegraphics[width=4.5in]{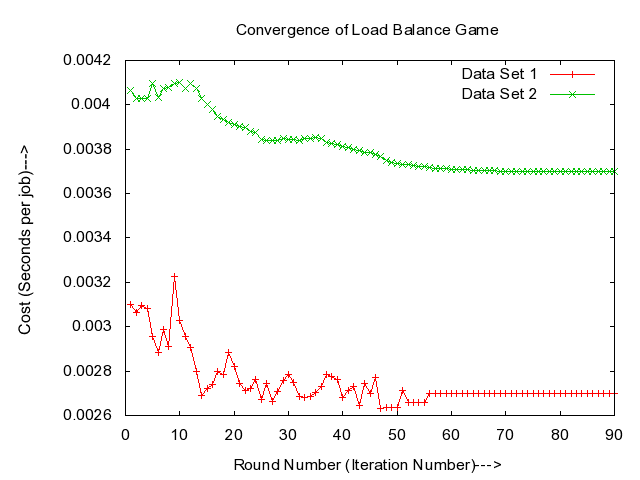}
    \caption{Convergence at Nash equilibrium}
    \label{fig:converge}
\end{figure}
In figure {\ref{fig:converge}} we have shown two test cases with different set of processing rates at servers and job generation rates at clients. We can see from the figure, that after certain number of rounds, the system converges to the equilibrium state, and at the equilibrium state the cost is stable and minimized. Note that Congestion game always have a pure strategy Nash Equilibrium, and hence the Nash Equilibrium here is unique.
\end{sloppypar}
\begin{sloppypar}
	\begin{figure}[!t]
    \centering
        \includegraphics[width=4.5in]{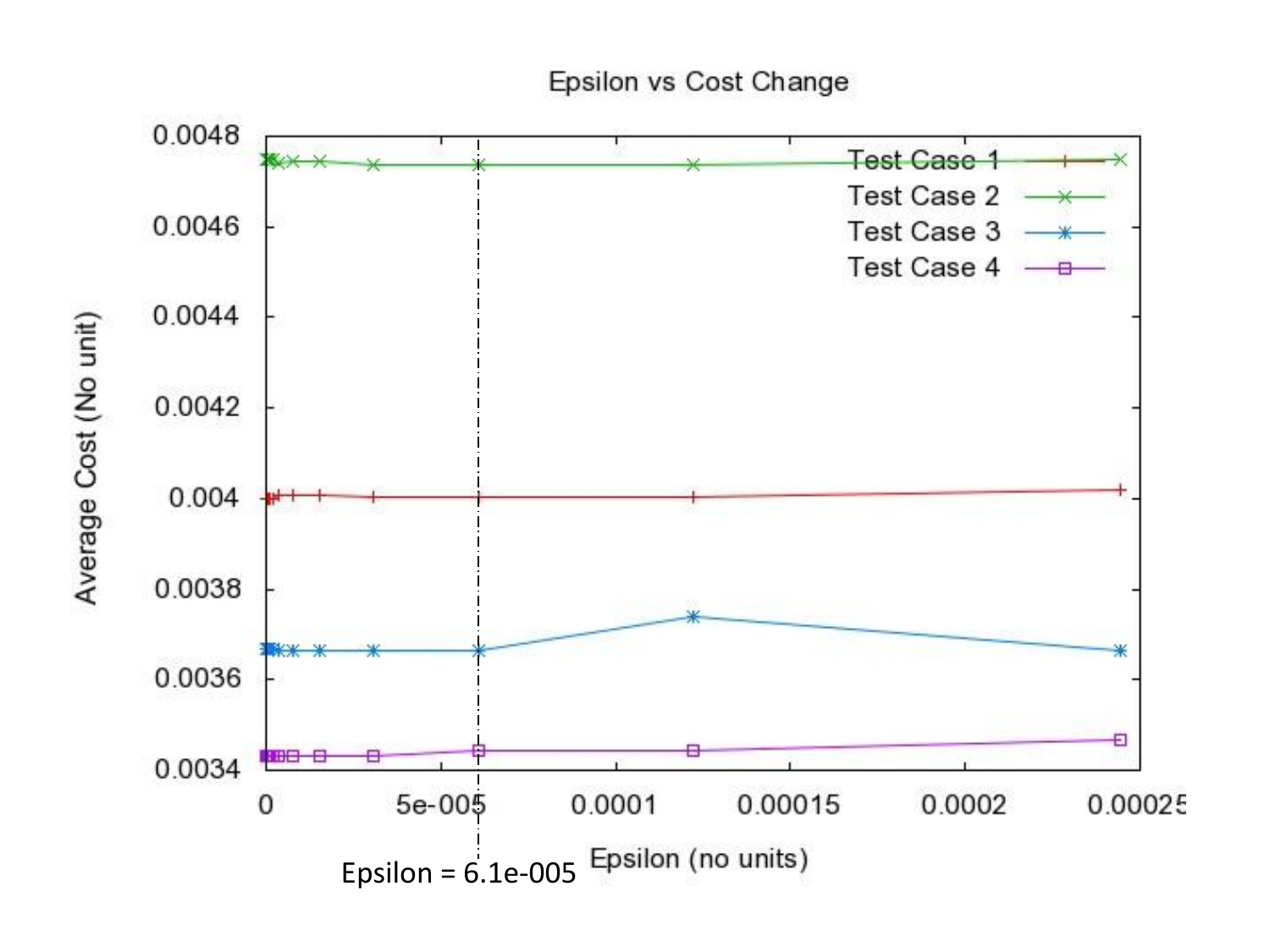}
    \caption{Cost with different $\epsilon$ values}
    \label{fig:cost}
\end{figure}
\begin{figure}[!t]
    \centering
        \includegraphics[width=4.5in]{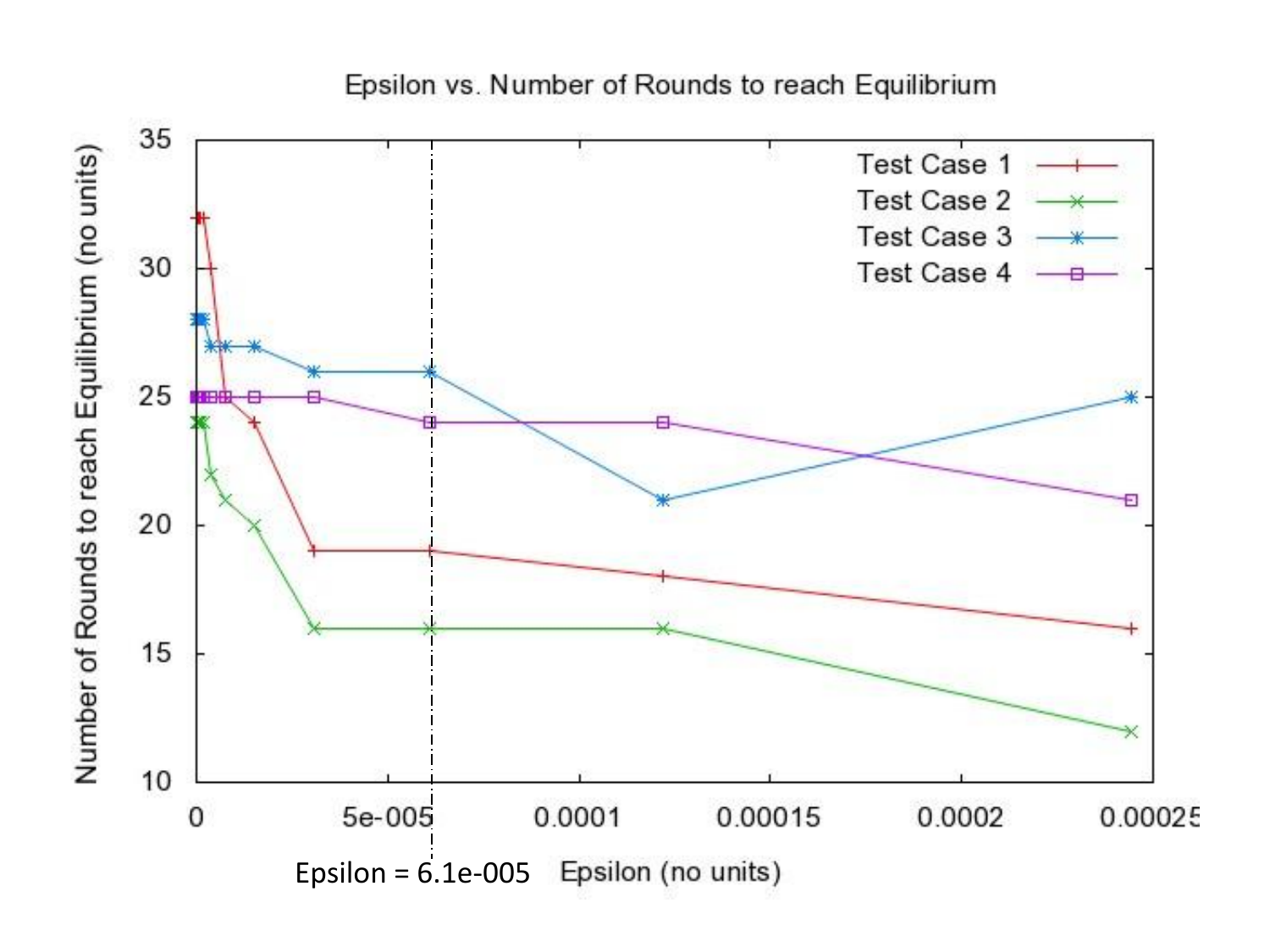}
    \caption{Number of rounds to find out the convergence state with different $\epsilon$ values}
    \label{fig:rounds}
\end{figure}
In figure {\ref{fig:cost}} and figure {\ref{fig:rounds}} we have shown four different test cases where each test case has same number of clients and servers, but with different set of maximum processing rates and job generation rates. From figure {\ref{fig:cost}}, we can see that at $\epsilon = 6.1\times{e^{-5}}$, all the systems converges to a state, where it gives minimum cost, and decreasing the value of $\epsilon$ does not improve the system cost. With $\epsilon = 6.1\times{e^{-5}}$, we can see from figure {\ref{fig:rounds}} that the number of rounds required to converge at Nash Equilibrium state is smaller than number of rounds required for pure strategy Congestion Game, which is essentially with $\epsilon = 0$. Thus we can see with properly chosen $\epsilon$ value, $\epsilon$-Congestion game converges to Nash Equilibrium rapidly, but with same cost than the pure Congestion Game. 
\end{sloppypar}
\begin{sloppypar}
	\begin{figure}[!t]
    \centering
        \includegraphics[width=4.5in]{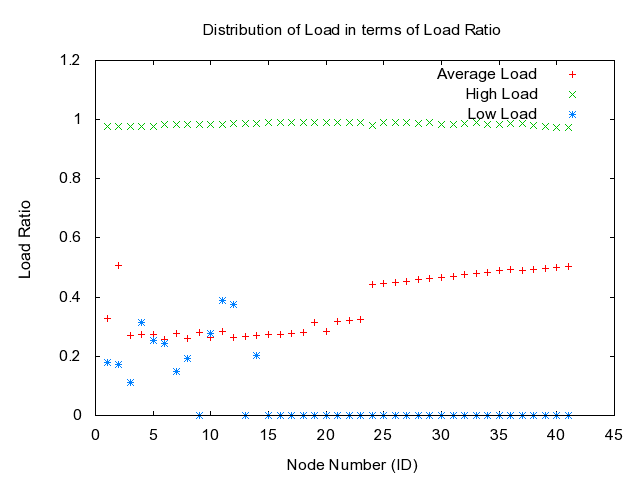}
    \caption{Distribution of Load in terms of Load Ratio}
    \label{fig:loadRatio}
\end{figure}
\begin{figure}[!t]
    \centering
        \includegraphics[width=4.5in]{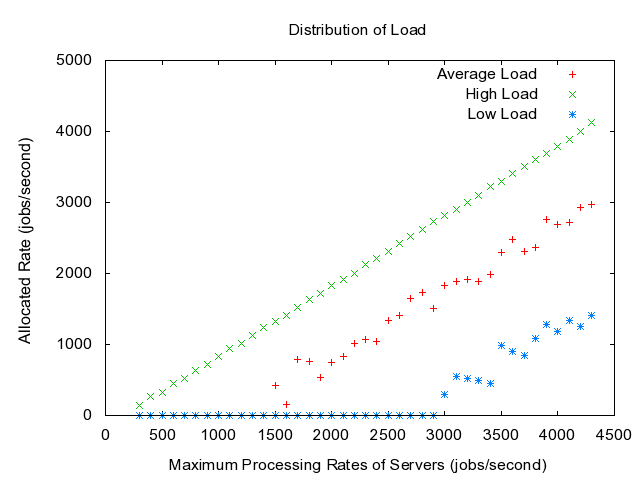}
    \caption{Distribution of Loads}
    \label{fig:load}
\end{figure}
In figure {\ref{fig:loadRatio}} and figure {\ref{fig:load}} we have shown the distribution of load among the server nodes. We have considered three cases, when the load is high, that is all clients generate jobs with high rates, when the load is low, that is job generation rate is low, and when the rate is average, that essentially follows the median of a Gaussian distribution. We can see from the figures that when load is high, all the server nodes has been assigned some load, but when load is low, only the servers with high processing rates gets the load to make the system performance better. The servers with low processing rates are only considered when load is very high. 
Hence we can have an effective distribution of loads among the processing nodes.  
\end{sloppypar}
 
\section{Conclusion and Future Works}
In this paper we have proposed a completely new framework for load balancing using $\epsilon$-Congestion game. We have shown the existence of pure Nash equilibrium in such a game and proposed a greedy algorithm to solve the problem. In spite of modeling the problem using pure congestion game as it is PLS-Complete, we have used the approximation of congestion game that converges to the equilibrium state within a finite number of steps. It should be noted that in {\cite{chien}}, the authors have also proved that even for symmetric congestion games with the bounded jump condition on all edges, finding a Nash equilibrium can still be PLS-complete; thus in this sense, bounded jumps are not a major restriction on the power of congestion games. This is a future motivation of study to check whether the processing overheads of $\epsilon$-congestion game is significantly larger than symmetric congestion game or not. Finally we have simulated the system to show that with a properly chosen value of $\epsilon$, the system converges rapidly to the Nash equilibrium than the pure Congestion Game. We have also shown by simulation that the system distributes the load properly among the processing nodes, and hence the load balancing is achieved using a distributed manner.

\end{document}